\documentclass[aps,prl,superscriptaddress,twocolumn]{revtex4-2}
\usepackage{amsmath,amsthm,xcolor} 
\usepackage[colorlinks=true,citecolor=blue,urlcolor=blue,breaklinks]{hyperref}
\usepackage[caption=false]{subfig}
\usepackage{amssymb}
\usepackage{graphicx}
\usepackage{here}
\usepackage{verbatim}
\newcommand{\mc}{\mathcal}

\usepackage{listings}
\def\<{\langle}\def\>{\rangle}
\newtheorem{theorem}{Theorem}

\newtheorem{proposition}[theorem]{Proposition}
 
\usepackage{pgfplots}
\pgfplotsset{compat=1.17}
\usetikzlibrary{backgrounds}

\newcommand{\updownarrows}{\mathbin{\rotatebox[origin=c]{90}{$\leftrightarrows$}}}

\usepackage[T1]{fontenc}

\begin{document}
	\title{Minimizing couplings in renormalization by preserving short-range mutual information}
	\author{Christian Bertoni}
	\affiliation{Dahlem Center for Complex Quantum Systems, Freie Universität Berlin, Germany}
	\email{bertoni@physik.fu-berlin.de}
	\author{Joseph M. Renes}
	\affiliation{Institute for Theoretical Physis, ETH Zürich, Switzerland}
	\email{renes@phys.ethz.ch}
	
	\begin{abstract}
		The connections between renormalization in statistical mechanics and information theory are intuitively evident, but a satisfactory theoretical treatment remains elusive.
		Recently, Koch-Janusz and Ringel proposed selecting a real-space renormalization map for classical lattice systems by minimizing the loss of long-range mutual information [Nat. Phys. 14, 578 (2018)]. The success of this technique has been related in part to the minimization of long-range couplings in the renormalized Hamiltonian [Lenggenhager et al., Phys. Rev. X 10, 011037 (2020)]. We show that to minimize these couplings the renormalization map should, somewhat counterintuitively, instead be chosen to minimize the loss of short-range mutual information between a block and its boundary. Moreover, the previous minimization is a relaxation of this approach, which indicates that the aims of preserving long-range physics and eliminating short-range couplings are related in a nontrivial way.

	\end{abstract}
	\maketitle
	
	Despite neither being able to experimentally probe nor theoretically precisely describe the microscopic details of the physical systems that surround us, via renormalization we are still able to make predictions and verify them to remarkable degrees of accuracy. A renormalization process progressively removes degrees of freedom from a physical system, mapping it to an effective system having the same physics at large scales \cite{wilson1,wilson2}. 
	One may regard the renormalization map as removing unimportant short-range information while leaving long-range information intact, and therefore 
	possible connections to information theory have been explored in several different approaches \cite{cfunction1,cfunction2,cfunction3,beny1,beny2,beny3,machta2013parameterspace}.
	One difficulty in the renormalization enterprise is finding an appropriate renormalization map. 
	In real space renormalization \cite{kadanoff}, 
	for example, there is no unique way to remove degrees of freedom, and a several maps can plausibly be used. Some work noticeably better than others \cite{swendsen1979}, but there is no clear criterion for choosing the best map. 
	
Recently, Koch-Janusz and Ringel \cite{Koch2017RSMI} proposed choosing real-space renormalization maps based on an information-theoretic criterion, as follows. 
Consider a spin model on a lattice $\Lambda$, and divide the lattice into non overlapping blocks $A_j$. 
Let $\mc R$ be a renormalization map on a single block, specifically a stochastic transformation on the random variables describing the spins in the block, and call its output on the $j$th block $A'_j$. 
In the renormalization procedure $\mc R$ is applied to each $A_j$, but here we need only focus on a single block $A$ with output $A'=\mc R(A)$. 
In particular, dividing the lattice into the block in question $A$, its neighbors within some distance $B$, and the remainder of the spins $C$, as illustrated in Figure~\ref{fig:3_regions}, Koch-Janusz and Ringel propose choosing 
\begin{align}
\label{eq:long}
\mc R_{\text{KJR}}=\operatorname{argmax}_{\mc R} I(A':C)_{\mc R(P)}\,,
\end{align}
where $P=\frac 1Z e^{-\beta H}$ is the Gibbs distribution of the spin system and $I(A:C)_P$ is the mutual information of random variables $A$ and $C$ under the distribution $P$. 

Due to the data processing inequality, it follows that $I(A:C)_P\geq I(A':C)_{\mc R(P)}$, and hence $\mc R_{\text{KJR}}$ retains the most mutual information between the block and the long range parts of the lattice. 
Koch-Janusz and Ringel argue that it therefore extracts the relevant degrees of freedom and that it results in a renormalized Hamiltonian with short-range couplings.
They also propose a machine-learning algorithm to determine $R_{\text{KJR}}$ on a parametrized subset of all possible maps. The resulting Real Space Mutual Information (RSMI) algorithm produces good results when benchmarked on various physical models. 
Lenggenhager et al.~\cite{leggenhager2020optimalRGfromIT} further showed that $\mc R_{\text{KJR}}$ does not create any long-range couplings within $C$ when $I(A:C)_P=I(A':C)_{\mathcal R(P)}$. Their theoretical work was expanded to field theory \cite{gokmen2021statistical} and their algorithm improved by using deep learning techniques \cite{gokmen2021phase}.
	
In this Letter we argue that, contrary to the above intuition, to minimize long-range couplings one should instead choose the renormalization map to retain \emph{short-range} mutual information: 
\begin{align}
\label{eq:short}
\mc R^\star=\operatorname{argmax}_{\mc R} I(A':B)_{\mc R(P)}\,.
\end{align}
As we show in detail below, in fact no map $\mc R$ can result in long-range couplings within $C$ or from $A$ to $C$, and $\mc R^\star$ additionally minimizes coupling within the boundary $B$. 
This approach has several other advantages.  
For one, the optimization is considerably simpler, as it only involves the block in question and its boundary. 
Moreover, it is the case that $I(A':B)_{\mc R(P)}\geq I(A':C)_{\mc R(P)}$ for every map $\mc R$, and hence the optimization in \eqref{eq:long} is a relaxation of the optimization in \eqref{eq:short}. We emphasize here that these two optimizations are born out of two different motivations: \eqref{eq:long} identifies the degrees of freedom that are most relevant to the long range physics, while \eqref{eq:short} aims to control the proliferation of couplings. It is not expected that these two motivations yield the same optimization problem, and the relaxation described above relates the two.
Finally, the optimizer of \eqref{eq:short} (as well as of \eqref{eq:long}) is a deterministic map, which makes brute-force optimization feasible for small blocks by searching the entire map space directly on the probability distribution, rather than by using sampling techniques. 
We illustrate how the optimization can be performed for $2\times 2$ maps using tensor network representations for the 2D Ising model.

\begin{figure}
\subfloat[Block and boundary
\label{fig:3_regions}]{%
\begin{tikzpicture}[x=0.55cm,y=0.55cm]
\fill[black!15,draw=black] (-2,-2) rectangle (2,2);
\fill[black!25,draw=black] (-1,-1) rectangle (1,1);
    \foreach \x in {-2.5,-1.5,-0.5,0.5,1.5,2.5} {
        \foreach \y in {-2.5,-1.5,-0.5,0.5,1.5,2.5} {
            \fill[color=black] (\x,\y) circle (0.09);
        }
    }
    \node at (0,0) {$A$};
    \node at (1.5,0) {$B$};
    \node at (2.5,0) {$C$};
\end{tikzpicture}
}\hspace{1cm}
\subfloat[A Markov network
\label{fig:Markov}]{
\begin{tikzpicture}[x=0.35cm,y=0.35cm]
\draw[step=.35cm,gray, thin] (-4.5,-4.5) grid (4.5,4.5);
\foreach \x in {-4,-3,...,4} {
        \foreach \y in {-4,-3,...,4} {
            \fill[color=white,draw=black] (\x,\y) circle (0.2);
        }
    }
\foreach \x/\y in {0/1, 1/1, 0/2, 1/2, 2/2, -2/-2} {
	\fill[color=black] (\x,\y) circle (0.2);
}
\foreach \x/\y in {0/0, -1/1, -1/2, 0/3, 1/3, 2/3, 3/3, 3/2, 2/1, 1/0, -2/-1, -2/-3, -1/-2,-3/-2,-1/-3,0/-3} {
	\fill[color=black!40,draw=black] (\x,\y) circle (0.2);
}
\end{tikzpicture}
}
\caption{a) Division of a 2D lattice system into the block to be renormalized $A$, its boundary $B$, and the rest of the lattice $C$. b) The random variables in the black region are conditionally independent of the those in the white region given the gray region, as the gray region shields the former from the latter in the Markov network. The regions need not be connected.}
\vspace{-0.6cm}
\end{figure}
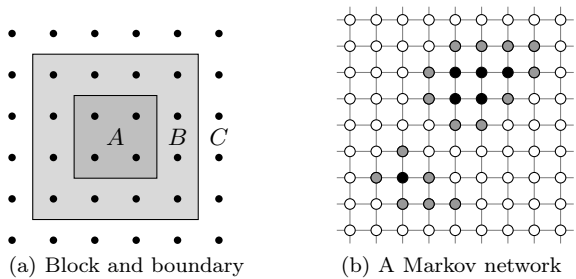

\textit{Gibbs states as Markov networks.---}
To prove our claims we make use of the Hammersley-Clifford theorem of probability theory, which states that every Gibbs state of a local Hamiltonian is a \emph{Markov network}. 
A Markov network is a (probability distribution on a) collection of random variables with  conditional independence relations that are captured by an undirected graph.
Consider a collection of random variables $V=(V_1,\dots,V_n)$ associated to vertices of a graph $\mathcal G$ and having a joint probability distribution $P(V)$. 
Vertices $V_j$ and $V_k$ connected by an edge in $\mathcal G$ correspond to dependent random variables, for which $I(V_j:V_k)\neq 0$.
Given three regions of the graph $A$, $B$, and $C$, corresponding to disjoint collections of the random variables, $B$ is said to shield $A$ from $C$ if all paths connecting $A$ to $C$ pass through $B$. 
The regions themselves need not be connected, as depicted in Figure~\ref{fig:Markov}. 

Then $(\mathcal G,P)$ is a Markov network if every two regions shielded by a third are conditionally independent, i.e.\ $A$ and $C$ are independent given the value of $B$. 
Put yet differently, the correlations between $A$ and $C$ are mediated entirely by $B$. 
Conditional independence can be succinctly expressed using the conditional mutual information (CMI) as $I(A:C|B)_P=0$, where 
\begin{align}
\label{eq:CMI}
I(A:C|B)_P:=I(A:BC)_P-I(A:B)_P\,.
\end{align}

The Hammersley-Clifford theorem~\cite{hammersley_markov_1971,koller2009probabilistic} then states that $(\mathcal G,P)$ is a Markov network if and only if $P(V)=e^{h(V)}$ for some local function $h$, meaning $h=\sum_{c\in \mathcal C} h_c$, where $\mathcal C$ is the set of cliques of the graph (the fully-connected subgraphs) and each $h_c$ is a function only of the variables involved in the clique $c$. 

The renormalization procedure begins with the Gibbs state of a local Hamiltonian $P\propto e^H$. 
Renormalizing a block $A$ with map $\mathcal R$ results in a new probability $P'=\mathcal R(P)=e^{h'}$, where we define $h'=\log P'$. 
Renormalizing all blocks results in some distribution $P''$, and the corresponding $h''$ is just the renormalized Hamiltonian, up to the inverse temperature $\beta$ and normalization constant factors. 
By the Hammersley-Clifford theorem, $h''$ will not contain any couplings between random variables which are conditionally independent, and this property can be established by showing that the CMI  vanishes. 
And by data processing, it is sufficient to consider just $h'$ to determine where new couplings may arise. 

\emph{Ruling out couplings.}---
The presence of the boundary $B$ around the block $A$ ensures that $\mc R$ creates no couplings  within $C$ nor from $A'$ to $C$. 
Consider two parts $C_1$ and $C_2$ of $C$ which are not already coupled. 
Thus they are conditionally independent given the remainder $R$ of the random variables comprising the system.
Region $A$ is a part of $R$, and the rest we can call $D$ so that $R=AD$. 
Since $B$ bounds $A$, it must be the case that $D$ shields $C_1$ from $C_2$ and therefore $I(C_1:C_2|D)_P=0$.
This does not change under application of any map $\mc R$, $I(C_1:C_2|D)_{\mc R(P)}=0$, and therefore $C_1$ and $C_2$ are not coupled in $h'$. To show the same thing, the authors of \cite{leggenhager2020optimalRGfromIT} prove instead that $I(C_1:C_2|A')=0$ by assuming that long range mutual information is preserved, i.e. $I(A:C)_P=I(A':C)_{\mathcal R(P)}$.
That $A'$ will not become coupled to anything in $C$ follows because all the correlations are mediated by $B$. 
Using the positivity of CMI and data processing, we have $0\leq I(A':C|B)_{\mathcal R(P)}\leq I(A:C|B)_P=0$. 

Hence, the main concern is couplings between parts of $B$ which may be induced by $\mc R$.
In one-dimensional systems, as depicted in Figure~\ref{fig:blockMI}, it turns out that coupling between $B_L$ and $B_R$ is related to the change in mutual information between the block $A$ and the boundary $B=B_LB_R$.  
If the mutual information is unchanged after $\mc R$, then $B_L$ and $B_R$ are uncoupled in $h'$. This is a consequence of the following more general statement.  
\begin{theorem}\label{thm:no_couplings}
	Consider a one-dimensional lattice model with nearest-neighbor Hamiltonian $H$ in a Gibbs state, divided into subregions as in  Figure \ref{fig:blockMI}. 
	For any renormalization map $\mc R: A\to A'$, $I(B_L:B_R|A')_{\mc R(P)}\leq I(A:B)_{P}-I(A':B)_{\mc R(P)}$. 
\end{theorem}
\begin{proof}
Start from $I(B_L:B_R|A')=I(B_L:B_RA')-I(B_L:A')$ and apply data processing to the first term to obtain $I(B_L:B_R|A')\leq I(B_L:B_RA)-I(B_L:A')$. 
Now note that $A$ and $C$ can be swapped in \eqref{eq:CMI}, i.e.\ $I(A:C|B)=I(C:AB)-I(C:B)$, and therefore $I(AB:C)-I(A:BC)=I(B:C)-I(A:B)$. 
Using this property for each term in the expression above gives $I(B_L:B_R|A')\leq I(A:B)-I(A':B)+I(B_R:B_LA')-I(B_R:A)$. 
Another application of data processing to the third term and the CMI definition gives 
$I(B_L:B_R|A')\leq I(A:B)-I(A':B)+I(B_R:B_L|A)$. 
The final term is zero by assumption.
\end{proof}

\begin{figure}
\centering
\begin{tikzpicture}[x=0.7cm,y=0.7cm]
\fill[black!25,draw=black] (-2,-0.5) rectangle (2,0.5);
\fill[black!15,draw=black] (-3,-0.5) rectangle (-2,0.5);
\fill[black!15,draw=black] (3,-0.5) rectangle (2,0.5);
\foreach \x in {-4.5,-3.5,-2.5,-1.5,-0.5,0.5,1.5,2.5,3.5,4.5} {
    \fill[color=black] (\x,0) circle (0.075);
    }
\node at (5.25,0) {$\cdots$};
\node at (-5.25,0) {$\cdots$};
\draw (-5.5,-0.5) -- (-3,-0.5) -- (-3,0.5) -- (-5.5,0.5);
\draw (5.5,-0.5) -- (3,-0.5) -- (3,0.5) -- (5.5,0.5);
\node at (4.5,0.9) {$C_R$};
\node at (-4.5,0.9) {$C_L$};
\node at (0,0.9) {$A$};
\node at (2.5,0.9) {$B_R$};
\node at (-2.5,0.9) {$B_L$};
\end{tikzpicture}
\caption{Division of a 1D lattice into block and boundary.
 \label{fig:blockMI}}
\end{figure}
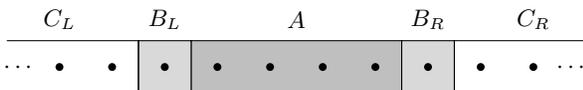

Typically, no nontrivial map $\mc R$ will precisely preserve the mutual information for reasons we shall explain in a moment. 
Nevertheless, minimizing the change in mutual information, by maximizing $I(A':B)_{\mc R(P)}$ as in \eqref{eq:short}, minimizes the coupling between $B_L$ and $B_R$. 
This is because the smaller the CMI, the closer the distribution $\mc R(P)$ is to some $P'$ in which $B_L$ and $B_R$ are conditionally independent, as measured by the total variational distance between distributions (see \cite[Lemma 1]{li_squashed_2018}). 
Hence smaller CMI leads to an associated $h'$ with weaker couplings.
Somewhat counterintuitively, then, to minimize couplings it is more important to preserve mutual information between a block and its boundary rather than between a block and distant spins.  

For isotropic systems, we can translate the $1D$ argument to multiple dimensions by treating a $D$ dimensional isotropic lattice as a $1D$ system in every direction, as proposed by Leggenhager et al.~\cite{leggenhager2020optimalRGfromIT}. 
The lattice can be separated into disconnected regions by hyperplanes creating effectively a $1D$ system (Figure \ref{fig:strips}) and the argument of Theorem~\ref{thm:no_couplings} carries over, so that no couplings will appear between the spins in the boundary strips $B_L$ and $B_R$. 
Couplings might still appear inside the central strip, but if the system is isotropic we can repeat the same argument with hyperplanes separating the renormalized block from the rest in a different dimension and expect that if a map maximized $I(A':B)$ in one dimension, it will do so also in the other dimension. This argument breaks down for non isotropic systems as the different directions may have different optimal maps.

Before proceeding to examine the two optimizations in more detail, let us remark that a renormalization map which precisely preserves the mutual information can actually be undone by a suitable stochastic map. 
This accords with the idea that no information is lost along the renormalization flow in this case by assumption, but one does not typically expect renormalization to be reversible.
Starting from $I(A':B)_{\mc R(P)}=I(A:B)_P$ and using the fact that $I(A:C|B)_P=I(A':C|B)_{\mc R(P)}=0$, it follows that the total mutual information is preserved, $I(A:BC)_P=I(A':BC)_{\mc R(P)}$. 
Then we can appeal to Lemma' 1 of \cite{li_squashed_2018}, which ensures that the so-called ``transpose'' map or Petz recovery map $\hat {\mc R}$ is such that $\hat {\mc R}\circ\mc R(P)=P$~\footnote{This gives a simple proof of the statement of \cite{leggenhager2020optimalRGfromIT}, as $I(A:C)=I(A':C)$ implies we can use data processing both from $A$ to $A'$ and from $A'$ to $A$. Thus, $I(C_1:C_2A')\leq I(C_1:C_2A)=I(C_1:A)\leq I(C_1:A')$ and hence $I(C_1:C_2|A')=0$. The equality is $I(C_1:C_2|A)=0$ and the inequalities are data processing.}. 
The transpose map depends on $\mc R$ and the marginal distribution of $A$ under $P$, but we shall not go into further details here.

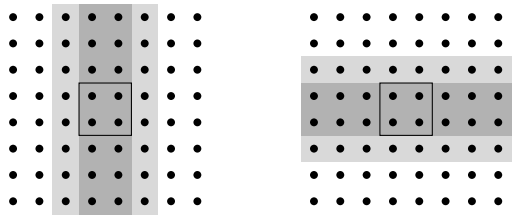
\begin{figure}
\begin{tikzpicture}[x=0.35cm,y=0.35cm]
\fill[black!15] (-2.5,-4.5) rectangle (1.5,3.5);
\fill[black!30] (-1.5,-4.5) rectangle (0.5,3.5);
\foreach \x in {-4,-3,...,3} {
        \foreach \y in {-4,-3,...,3} {
            \fill[color=black] (\x,\y) circle (0.15);
        }
    }
\draw (-1.495,-1.495) rectangle (0.495,0.495);

\begin{scope}[xshift=4cm]
\fill[black!15] (-4.5,-2.5) rectangle (3.5,1.5);
\fill[black!30] (-4.5,-1.5) rectangle (3.5,0.5);
\foreach \x in {-4,-3,...,3} {
        \foreach \y in {-4,-3,...,3} {
            \fill[color=black] (\x,\y) circle (0.15);
        }
    }
\draw (-1.495,-1.495) rectangle (0.495,0.495);
\end{scope}
\end{tikzpicture}
\caption{The dark and light gray strips indicate the blocks that are used when treating the system as one dimensional in each direction, while the square indicates a block to be renormalized. 
If the renormalization map is optimal, the light gray strips are uncoupled. If the system is isotropic, the optimal maps for the two directions are the same. }
	\label{fig:strips}
	\vspace{-.5cm}
\end{figure}

\textit{Optimization.}---
Computing $I(A:B)$ does not require handling the whole probability distribution, but only the marginal distribution on the $AB$ subsystem. 
This simplifies the optimization relative to Koch-Janusz and Ringel's proposal, where the distribution on the entire spin system must be treated somehow. 
As mentioned above, \eqref{eq:long} is a relaxation of \eqref{eq:short} in that $I(A':C)_{\mc R(P)}\leq I(A':B)_{\mc R(P)}$. This follows directly from the definition of the CMI and the Markov condition: $I(A':BC)=I(A':B)$ since $I(A':C|B)=0$, but then $I(A':C)\leq I(A':B)$ by data processing. 
The equality $I(A':C|B)=0$ reflects the fact that all correlations between $A'$ and $C$ are mediated through $B$. Therefore, maximizing the mutual information of the former sets a lower bound on the mutual information of the latter.   

In both \eqref{eq:long} and \eqref{eq:short} the optimal map $\mc R^\star$ is necessarily deterministic, i.e.\ all its transition probabilities are either zero or one.    
This follows because the objective function, the mutual information, is a convex function of the optimization variable, the map $\mc R$, and the extreme points of stochastic maps are deterministic maps.
\begin{proposition}
	Let $\mc{C}$ be the space of channels from $A$ to $A'$. For a fixed probability distribution $P_{AB}$ the function $\mc{C}\rightarrow \mathbb{R}_+$, $W\mapsto I(A':B)_{W(P)}$ is convex.
\end{proposition}
\begin{proof}
Consider a collection of channels $\{W_z\}_{z\in \mc{Z}}$ indexed by the values of a finite random variable $Z$ with distribution $Q$. 
The average channel $W_Z$ is just $W_Z(P_{AB})=\sum_{z\in \mc{Z}}Q(z)W_z(P_{AB})$ for any $P_{AB}$, leading to mutual information $I(A':B)_{W_Z(P)}$. 
For simplicity, denote $W_Z(P)$ just by $P'$. 
Meanwhile, the average mutual information is given by the CMI $I(A':B|Z)_{P'}$ since
	\begin{align}
	&\sum_{z\in \mc{Z}} Q(z) I(A':B)_{W_z(P)}\nonumber \\
	&=\sum_{z\in \mc{Z}} Q(z) I(A':B|Z=z)_{W_Z(P)}
	=I(A':B|Z)_{P'}\,.
	\end{align}
But then, since $B$ and $Z$ are uncorrelated, we obtain
	\begin{align}
	I(A':B|Z)_{P'}
	&=I(A'Z:B)_{P'}-I(Z:B)_{P'}\\
	&= I(A'Z:B)_{P'}\geq I(A':B)_{P'}\,,
	\end{align}
	and therefore the mapping is convex. 
\end{proof}
When maximizing a convex function over a convex set, the optimum will occur at one of the extreme points~\cite[Theorem 32.2]{rockafellar_convex_1970}, which in this case are the deterministic maps~\cite[Theorem 1]{davis_markov_1961}. 
This simplifies the optimization by making the search space finite. While brute force might still be out of reach for interesting systems, more sophisticated methods such as machine learning techniques can be informed by this fact.

\textit{The Ising model.}--- Consider renormalization maps on $2\times 2$ blocks in the 2D square-lattice Ising model.
To investigate which maps are optimal according to \eqref{eq:short}, we use the Corner Transfer Matrix algorithm \cite{nishino_corner_1996} to extract the marginal distribution of a $4\times 4$ block, and we measure the change in mutual information between the central $2\times 2$ block and its boundary after each of the possible $2^{16}$ deterministic maps mapping this block to a single spin. 
We then compute the change in mutual information for each map over the range of temperatures $\beta\in [0.1\beta_c, 1.9\beta_c]$ and find the optimal map at each temperature. 
In Figure \ref{fig:optcomparison} we show the change in mutual information compared with the minimum value for some common maps:
\begin{enumerate}
	\item Decimation: the value of the renormalized spin is simply the value of one of the $4$ spins in the block.
	
	\item Majority vote: the renormalized spin is assigned a value $+1$ if the majority of the spins in the block are $+1$, and vice versa. Ties must be broken with a $2\times 2$ block, we do this in 4 possible ways: using a predetermined fixed value (i.e.\ the ties are always resolved with $+1$ or $-1$), using one of the spins in the block (hence the map becomes decimation in case of ties), or choosing a value at random.
	
	\item Biased: the all configurations are mapped to $+1$ except for $(-1,-1,-1,-1)$ or, vice versa, to $-1$ except for $(+1,+1,+1,+1)$. 
\end{enumerate}

Some of these maps are not symmetric under spin flips, namely the majority vote with fixed value tie breaker and the biased maps. Which version is optimal depends on the symmetry breaking low temperature state that has been selected during the simulation. 
We call the tie breaker or the biased map ``aligned'' (denoted $\upuparrows$ in the figure) if the relevant fixed value for the renormalized spin is aligned with the magnetization in the symmetry-breaking state, and ``antialigned'' ($\updownarrows)$ otherwise. 

At high temperature $(\beta/\beta_c\lesssim 0.3554)$, the optimal map is decimation, afterwards, for $0.3554\lesssim\beta/\beta_c \lesssim0.6109$, majority vote with tie breaks decided by decimation. From that point up to the critical temperature, both version of fixed tie breaker majority vote are optimal, the aligned version remains so up to $\beta/\beta_c\approx 1.0509$, after which the low temperature symmetry breaking prevails and the best map is the aligned biased map. 

Interestingly, majority vote with random tie breaker is rather far from optimal (it cannot be optimal as it is not deterministic) and fares worse of all other tie breakers except the antialigned one at low temperature. 
It can also be seen that decimation performs poorly, especially around the critical point. 
This is consistent with the observations of \cite{swendsen1979}. 

\pgfplotsset{width=9cm,height=7cm,compat=1.16}

\definecolor{mblue}{rgb}{0.368417, 0.506779, 0.709798}
\definecolor{morange}{rgb}{0.880722, 0.611041, 0.142051}
\definecolor{mgreen}{rgb}{0.560181, 0.691569, 0.194885}
\definecolor{mred}{rgb}{0.922526, 0.385626, 0.209179}
\definecolor{mpurple}{rgb}{0.528488, 0.470624, 0.701351}

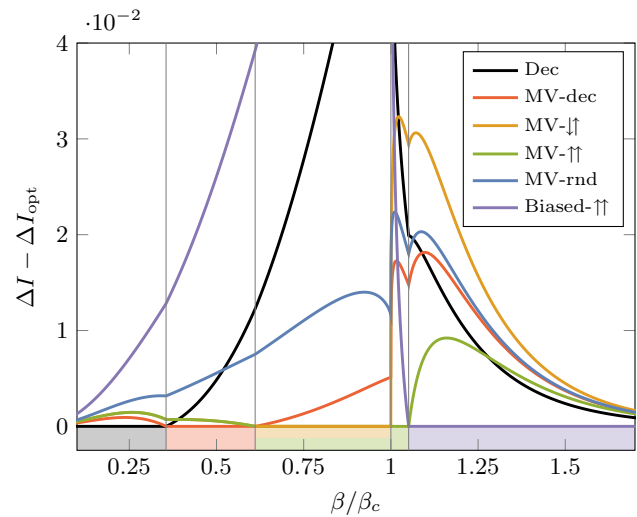
\begin{figure}
\centering
\begin{tikzpicture}
\begin{axis}[
enlargelimits = false,
scaled x ticks = false,
xlabel={$\beta/\beta_c$},
ylabel={$\Delta I-\Delta I_{\rm opt}$},
xtick={0,0.25,.5,.75,1,1.25,1.5,2},
xmin=0.1, xmax=1.7, ymin=-0.0025,ymax=0.04,
legend entries={Dec,MV-dec,{MV-$\updownarrows$},{MV-$\upuparrows$},MV-rnd,{Biased-$\upuparrows$}},
legend style={font=\scriptsize},
legend cell align=left,
no markers,
every axis plot/.append style={line width=1.1pt}]
\addplot [black] table[y index=5] {datar.dat};
\addplot [mred] table[y index=3] {datar.dat};
\addplot [morange] table[y index=2] {datar.dat};
\addplot [mgreen] table[y index=1] {datar.dat};
\addplot [mblue] table[y index=4] {datar.dat};
\addplot [mpurple] table[y index=6] {datar.dat};
\addplot [mark=none,morange] coordinates {(0.6109,0) (1,0)};

\addplot [mark=none,thin,gray] coordinates {(.3554, -0.0025) (.3554, .04)};
\addplot [mark=none,thin,gray] coordinates {(.6109, -0.0025) (.6109, .04)};
\addplot [mark=none,thin,gray] coordinates {(1, -0.0025) (1, .04)};
\addplot [mark=none,thin,gray] coordinates {(1.0509, -0.0025) (1.0509, .04)};

\begin{scope}[on background layer]
\fill[black!20] (0.1,-0.0025) rectangle (.3554,0);
\fill[mred!30] (.3554,-0.0025) rectangle (0.6109,0);
\fill[mgreen!30] (0.6109,-0.0025) rectangle (1,-0.00125);
\fill[morange!30] (0.6109,-0.00125) rectangle (1,-0.0);
\fill[mgreen!30] (1,-0.0025) rectangle (1.0509,0);
\fill[mpurple!30] (1.0509,-0.0025) rectangle (1.7,0);
\end{scope}
\end{axis}
\end{tikzpicture}
\caption{Difference of the mutual information change for each map above the optimal change, as a function of inverse temperature. 
Each shaded region indicates which map is optimal in the corresponding interval. 
Note that while both majority vote maps which break ties aligned (MV-$\upuparrows$) and antialigned (MV-$\updownarrows$) with the overall magnetization are optimal in the interval (0.6109,1), the random tiebreaker map  (MV-rnd) is far from optimal. \label{fig:optcomparison}}

\end{figure}

\emph{Conclusions.---} In this Letter, we argued that maximizing the short-range mutual information between a block and its boundary yields a renormalized system with reduced long-range couplings.
In particular, couplings are never introduced beyond the boundary region of the renormalization map, and are suppressed when more of the short-range mutual information is preserved. 
This gives an information-theoretic account of some aspects of renormalization.
The optimization suggested by this approach leads to a simple brute-force algorithm for finding the optimal renormalization map which requires only the probability distribution of the input region of the map and its boundary. It is efficient enough for small systems, as demonstrated in the 2D Ising model. 
Further work is required to explore the robustness of this result when information is only approximately preserved, perhaps by using an approximate generalization of the Hammersely-Clifford theorem.

Our approach contrasts with the focus of \cite{Koch2017RSMI} and \cite{leggenhager2020optimalRGfromIT}, which maximizes the long-range mutual information with the dual goals of capturing the relevant degrees of freedom and reducing long-range couplings. 
The fact that their long-range mutual information optimization is a relaxation of our short-range optimization implies some connection between these goals: If we view extracting the relevant information as the primary justification for the long-range optimization (an intuitively very plausible statement), then it will necessarily do this by minimizing long-range couplings in the renormalized Hamiltonian to some extent. 
The open question is how much. 
It would therefore be interesting to investigate under what conditions or in which models the optimal renormalization maps of the two approaches actually coincide. 
To this end it would also be interesting to modify the RSMI algorithm to focus on short-range mutual information, as exact optimization is computationally difficult for more complicated models. 
In either scenario one may also be able to take into account the fact that the optimal renormalization map is necessarily deterministic. 

Finally, it should be noted that the focus on short-range versus long-range information here is reminiscent of the relation between the Tensor Renormalization Group (TRG) \cite{LevinTRG2007} and the Tensor Network Renormalization (TNR) \cite{EvenblyTNR2015} algorithms. 
The latter is a refinement of the former in which the additional steps are meant to remove short-range correlations, improving the algorithm near the critical point. 
Here the setting is block-spin renormalization, i.e.\ maps on the physical degrees of freedom and not the tensors in the tensor-network description, but again the focus is on the short-range couplings. 
It would be interesting to investigate if information-theoretic methods can be used to give tensor network algorithms.  

\emph{Acknowledgments.---} We thank Doruk Efe Gökmen and Maciej Koch-Janusz for useful discussions. CB acknowledges support from the Deutsche Forschungsgemeinschaft via grant CRC183, JMR the Swiss National Science Foundation via the National Center for Competence in Research for Quantum Science and Technology (QSIT). 
	\bibliography{bibliography}

	\end{document}